\documentclass[11pt]{amsart}

\usepackage{graphicx}
\usepackage{dcolumn}
\usepackage{bm}
\usepackage{mathrsfs}
\usepackage{amsrefs}
\usepackage{hyperref}
\usepackage{amssymb}
\usepackage{amsaddr}
\usepackage{amsmath}
\usepackage{amssymb}
\usepackage{enumitem}
\usepackage{dsfont}

\usepackage{xcolor} 

\newcommand{\mcal}[1]{\mathcal{#1}}

\newcommand{\mbb}[1]{\mathbb{#1}}

\newcommand{\veps}{\varepsilon}

\newcommand{\N}{\mbb{N}}

\newcommand{\norm}[1]{\left\|#1\right\|}
\newcommand{\abs}[1]{\left|#1\right|}

\newtheorem{thm}{Theorem}
\newtheorem{cor}[thm]{Corollary}
\newtheorem{lem}[thm]{Lemma}
\newtheorem{prop}[thm]{Proposition}
\theoremstyle{definition}
\newtheorem{ex}{Example}

\begin{document}

\title[Semigroup approximations with non-uniform...]{Chernoff's product formula: Semigroup approximations with non-uniform time intervals}
\author{J. Z. Bern\'ad}
\address{Forschungszentrum J\"ulich, Institute of Quantum Control,
Peter Gr\"unberg Institut (PGI-8), 52425 J\"ulich, Germany}
\email{j.bernad@fz-juelich.de}
\author{A. B. Frigyik}
\address{\'Obuda University, 1081 Budapest, Hungary}
\email{frigyik.andras@bgk.uni-obuda.hu}

\date{\today}

\begin{abstract}
Often, when we consider the time evolution of a system, we resort to
approximation: Instead of calculating the exact orbit, we divide the time
interval in question into uniform segments. Chernoff's results in this direction
provide us with a general approximation scheme. There are situations when we
need to break the interval into uneven pieces. In this paper, we explore
alternative conditions to the one found by Smolyanov {\it et al.} such that Chernoff's original result can be extended to unevenly distributed time intervals. 
Two applications concerning the foundations of quantum mechanics and the central limit theorem are presented.
\end{abstract}

\maketitle

\section{Introduction}\label{sec1}

Many physically relevant 
processes are modeled by evolution equations.  
These evolution equations may
contain terms with different mathematical properties. While each term can be analyzed and understood separately, it  
may not be clear how the 
entire evolution takes place. Typical examples are the Feynman-Kac formula and Feynman path integrals in quantum mechanics \cite{Johnson}. Perturbation and approximation are some of the standard methods we turn to if we want to obtain the complete evolution. 
Quite often, we can associate a semigroup to an evolution equation and in that case what we would like to obtain is the generator of that semigroup. 
An important milestone of this project was Trotter's seminal paper \cite{Trotter}, which later was generalized by Chernoff \cite{Chernoff} presenting a product formula capable of reproducing many of the exponential formulas of semigroup theory. This product formula has been the subject of many further 
investigations. Some improved its rate of convergence (e.g. \cite{Zagrebnov1}), others discussed its behavior in the operator norm 
topology (e.g. \cite{Zagrebnov2}), while yet others extended 
its use to 
the non-autonomous differential equations (e.g. \cite{Vuillermot}) or to a framework with mean ergodic theorems \cite{Bernad}. In this
article, we focus on an extension of Chernoff's result to the situation where
the evolution steps of the terms in the product are not equidistant. Such an extension requires slightly stronger assumptions to obtain the same statement as Chernoff's
and therefore first we review the results of Refs. \cites{Smolyanov1,Smolyanov2}. Then, we provide a different approach. Finally, we demonstrate the usefulness of these types of results through a couple of examples from quantum mechanics and probability theory.

\section{Approximation with non-uniform time intervals}

In paper \cite{Chernoff}, Chernoff states the following theorem (see also \cite{engel1999one}): 

\begin{thm}
\label{Thm:Chernoff}
  Suppose $X$ is a Banach space and there is a function
  \begin{equation}
    V:[0,\infty) \to \mcal{B}(X), \nonumber
 \end{equation}
from the non-negative reals into the bounded linear operators on $X$. Function
$V$ satisfies the following conditions: it is power bounded, i.e. there is an $M \geq 1$
such that $\norm{[V(t)]^m}\leq M$ for all $m\in \N$ and for each $t \geq 0$, and  $V(0)=I$, where $I$ is the identity operator on $X$.

If it is true that
\begin{equation}
  Ax := \lim_{h\to 0} \frac{V(h)x-x}{h} \nonumber
\end{equation}
exists for all $x\in D(A)\subset X$, where $D(A)$ is a dense subspace of $X$ such
that, for some $\lambda_0 > 0$ the subspace $(\lambda_0-A)D(A)$ is dense in $X$, as
well, then the closure $\overline{A}$ of $A$ generates a strongly
continuous semigroup $(T(t))_{t\geq 0}$. For every $x\in X$
\begin{equation}
  T(t)x = \lim_{n\to \infty} V^n\left(t/n\right)x, \label{eq:Chernoff}
\end{equation}
where the limit is uniform on $[0,t_0]$ for any given $t_0\in [0,\infty)$.
\end{thm}

What this theorem says is that we can approximate the action of the semigroup
$T$ by dividing the time over which the semigroup 
acts into uniform pieces and
approximate the action on each one by a simpler action.

Do we really need to keep the division uniform? Consider the very simple case of
the exponential function:
\begin{equation}
  e^t =\lim_{n\to \infty} \left(1+\frac{t}{n}\right)^n.
\end{equation}
In general, if $\left(a_n\right)_{n\in \N}$ is a sequence such that $a_n\to
\infty$ 
as $n\to \infty$ then
\begin{equation}
  e^t =\lim_{n\to \infty} \left(1+\frac{t}{a_n}\right)^{a_n}.
\end{equation}

On the other hand, we can think about this approximation in terms of the
corresponding interval: Consider the interval $[0,t]$, or for simplicity let 
the interval be $[0,1]$,
and divide it into $n$ equal non-overlapping intervals that cover the whole of
it. The approximation above corresponds roughly (exactly, if all the $a_n$ are
positive integers) to this particular partition: The length of one step is $\frac{1}{a_n}$, but the division is still uniform in the sense that each subinterval has the same length. Instead of
a uniform partitioning like that, one can consider an even more general division and the corresponding limiting product formula. This problem, though, seems to be a generalization, it holds great importance in several stroboscopic processes, like the quantum Zeno effect or dynamical decoupling \cite{Facchi}, where in experimental settings the uniform division of time is not realistic.

\begin{prop}
 \label{Prop:exp}
 For each $n\in \N$, let $\left(a_{n,i}\right)_{i=1}^n$ be a sequence of positive numbers such
that $\sum_{i=1}^n a_{n,i}=1$. If, in addition, we assume that
\begin{equation}
  \lim_{n\to \infty} \max_{i} a_{n,i} =0, \nonumber
\end{equation}
then
\begin{equation}
  e^t = \lim_{n\to \infty} \prod_{i=1}^n \left(1+a_{n,i} t\right). \label{propprove}
\end{equation}
\end{prop}

The proof of this statement can be found, for example, Lemma $4.8$ in
\cite{kallenberg2002foundations}, where these kind of sequences are called "null
arrays". 

Application of this approach to semigroups immediately runs into a
difficulty. The original proof relies heavily on commutativity of certain products, so for any $t\in [0,\infty)$ we have 
\[
    V(t)^mV(t)=V(t)V(t)^m,
\]
for any $m\in\N$. If $\{t_i\}_{i=0}^m$ are all different, i.e, if $t_i=t_j$ iff $i=j$, then still we have 
\[
    V(t_0)\cdot \prod_{i=1}^m V(t_i)=\prod_{i=0}^{m-1}V(t_i)\cdot V(t_m),
\]
but we cannot claim that the latter product is 
\[
    V(t_m)\cdot\prod_{i=0}^{m-1}V(t_i).
\]

We lose the ability to rearrange the factors, which is the crux of the original proof. 


In \cite{Smolyanov1}, the authors have proved the following statement and demonstrated it in Proposition $3$ of 
the same paper:
\begin{prop}
 \label{Prop:Smolyanov}
 Let $A$ be the generator of a strongly continuous semigroup on a Banach space $X$. Let $ V:[0,\infty) \to \mcal{B}(X)$ be a one-parameter family of contractions on $X$ such that
 \begin{equation}
 \label{eq:speccond}
 V(0)=I \quad \text{and} \quad \lim_{h \to 0}
 \frac{V(h)-I}{h}e^{aA} x = A e^{aA}x,
 \end{equation}
for all $a>0$ and $x\in D(A)$. Let the sequence
$\left(a_{n,i}\right)_{i=1}^n$  satisfy the assumptions in Proposition \ref{Prop:exp}. Then,
\begin{equation}
 \lim_{n \to \infty} \norm{\prod^n_{i=1} V(a_{n,i}t)x-e^{tA} x}=0,   
\end{equation}
for all $x \in X$.
\end{prop}

In order for the above proposition to work, we need an immediate assumption that $A$ generates a strongly continuous semigroup and fulfills the condition in \eqref{eq:speccond}. This is closely related to the Lax Equivalence Theorem \cite{Lax}, because if $V$ satisfies \eqref{eq:speccond} then it is called consistent finite difference scheme.

However, if we want to demonstrate convergence without assuming in advance the generator property of $A$ we need other assumptions. Our strategy relies on the use of the Second Trotter-Kato Approximation Theorem (\cite{engel1999one}, Chapter III $4.9$). This approach creates a greater hurdle, and in order to mitigate this, we introduced a couple of assumptions. We consider $V(t)V(t')=V(t')V(t)$ for all $t, t' \geq0$, which does not imply that $V(t)$ is a semigroup (see Examples \ref{ex1} and \ref{ex2}).
Furthermore, we decided to work with contractions. We will need a stronger assumption for the
array $\left(a_{n,i}\right)_{i=1}^n$, $n\in \N$, as well. It should not deviate
from the uniform distribution too much:

\begin{equation}
  \lim_{n\to \infty} \sum_{i=1}^n \abs{\frac{1}{n} -a_{n,i}}=0.\label{eq:absdev}  
\end{equation}

\begin{lem}\label{lem:stronger}
  Every sequence of positive numbers $\left(a_{n,i}\right)_{i=1}^n$, $n\in \N$ with $\sum_{i=1}^n a_{n,i}=1$ that satisfies
  Eq.\eqref{eq:absdev} also satisfies the condition in Proposition
  \ref{Prop:exp}. 
\end{lem}

The converse is not true. Let $\left(a_{n,i}\right)_{i=1}^n$, $n\in \N$ be the
array defined the following way.
\[
  a_{n,i}=
  \begin{cases}
    \frac{1}{n} + (-1)^i \frac{1}{2n} & \text{if $n$ is even},\\
    \hfil \frac{1}{n}\hfil & \text{if $n$ is odd}.
  \end{cases}
\]

It is true that
\[
  \lim_{n\to \infty} \max_{i} a_{n,i} =0
\]
but 
\[
  \lim_{n\to \infty} \sum_{i=1}^{2n} \abs{\frac{1}{2n} -a_{2n,i}}=\frac{1}{2}.
\]

The proof of Lemma \ref{lem:stronger} is straightforward:
\begin{proof}
  Since
  \[
    \lim_{n\to \infty} \sum_{i=1}^n \abs{\frac{1}{n} -a_{n,i}}=0
  \]
  for every $\veps >0$ there is an $N\in \N$ such that if $n\geq N$ then
  $\frac{1}{n} < \veps$ and
  \[
    \sum_{i=1}^n \abs{\frac{1}{n} -a_{n,i}} <\veps.
  \]
  This implies that if $n\geq N$ then for all $1\leq i\leq n$ we have
  $\abs{\frac{1}{n} -a_{n,i}} <\veps$ and therefore $0\leq a_{n,i}<
  2\veps$. Consequently $\max_{i} a_{n,i} <2\veps$ and in the limit
  \[
    \lim_{n\to \infty} \max_{i} a_{n,i} =0.
  \]
  
\end{proof}

The following Lemma contains the precise formulation of the assumptions
that were mentioned above.

\begin{lem}
\label{Lemma}
Let $ V:[0,\infty) \to \mcal{B}(X)$ and the sequence
$\left(a_{n,i}\right)_{i=1}^n$  satisfy the assumptions in Theorem
\ref{Thm:Chernoff}, Proposition \ref{Prop:exp} and Eq.\eqref{eq:absdev},
respectively. Suppose $V(t)$ is a contraction for $t\geq 0$ and $V(t)V(t')=V(t')V(t)$ for all $t, t' \geq0$.  In this case,
\begin{equation}
 \lim_{n \to \infty} \norm{V^n(t/n)x - \prod^n_{i=1} V(a_{n,i}t)x}=0 \label{eq:extra} 
\end{equation}
 for all $x \in X$ and uniformly for $t \in [0,t_0]$.
\end{lem}

\begin{proof}

  Let $n\in \N$ and consider the norm of the difference in \eqref{eq:extra}:
  \begin{align*}
    &\norm{V^n(t/n)x - \prod^n_{i=1} V(a_{n,i}t)x}= \Bigg\| V^n(t/n)x - V^{n-1}(t/n)V(a_{n,n}t)x 
     \\
     &\left. \qquad + V^{n-1}(t/n)V(a_{n,n}t)x - \prod^n_{i=1} V(a_{n,i}t)x\right\|\\
    &\leq \norm{V(t/n)x - V(a_{n,n}t)x} + \norm{ \left(
      V^{n-1}(t/n) - \prod^{n}_{i=2} V(a_{n,i}t)\right)V(a_{n,n}t)x} \\
    &\leq  \sum^n_{i=1} \norm{\Big( V(t/n) - V(a_{n,i}t)\Big) x},
  \end{align*}
  where we used the commutativity of $V(t)$ and the fact that they are contractions. The last inequality is obtained by induction. 

 For $s>0$, we define
 \begin{equation}
     A_s:=\frac{V(s)-I}{s} \in \mcal{B}(X),
 \end{equation}
where $A_s x \to Ax$ for all $x \in D(A)$ as $s\to 0$. Then,
\begin{eqnarray}
  &&\norm{V(t/n)x - V(a_{n,i}t)x} = t\norm{\frac{1}{n}A_{t/n}x-a_{n,i} A_{a_{n,i}t}x} \nonumber \\
  &&\leq t \cdot \abs{\frac{1}{n}-a_{n,i}} \norm{A_{t/n}x} + t \cdot a_{n,i} \cdot \norm{A_{t/n}x-A_{a_{n,i}t}x} 
\end{eqnarray}
and observe that
\begin{eqnarray}
 &&\sum^n_{i=1} \norm{\Big( V(t/n) - V(a_{n,i}t)\Big) x} \leq t \cdot \norm{A_{t/n}x} \sum_{i=1}^n \abs{\frac{1}{n}-a_{n,i}} \nonumber \\
 &&+ t \cdot \underbrace{\sum^n_{i=1} a_{n,i}}_{=1} \cdot \Big( \norm{A_{t/n}x-Ax} + \norm{A_{a_{n,i}t}x-Ax}\Big). 
\end{eqnarray}
Since
  \begin{equation}
    \lim_{n\to \infty} \sum_{i=1}^n \abs{\frac{1}{n}-a_{n,i}} =0,
  \end{equation}
and hence 
\[
  \lim_{n\to \infty} \max_{i} a_{n,i} =0,
\]
the statement is true for all $x\in D(A)$ and uniformly for $t \in [0,t_0]$. However, 
\begin{equation}
\norm{V^n(t/n) - \prod^n_{i=1} V(a_{n,i}t)} \leq 2
\end{equation}
and thus the Banach-Steinhaus theorem (\cite{Dunford}, Theorem II.1.18) gives the desired conclusion. 
\end{proof}

\section{Main result}

Our main result is the following.

\begin{thm}
\label{Thm:main}
Let $ V:[0,\infty) \to \mcal{B}(X)$ be a contraction satisfying the assumptions in Theorem
\ref{Thm:Chernoff} and $V(t)V(t')=V(t')V(t)$ for all $t, t' \geq0$. Let $\left(a_{n,i}\right)_{i=1}^n$, $n\in \N$  be a sequence of positive numbers such that $\sum_{i=1}^n a_{n,i}=1$ and
\begin{equation}
  \lim_{n\to \infty} \sum_{i=1}^n \abs{\frac{1}{n} -a_{n,i}}=0.  
\end{equation}
Then
\begin{equation}
  T(t)x = \lim_{n\to \infty}  \prod^n_{i=1} V(a_{n,i}t)x, \label{eq:Chernoff2}
\end{equation}
for all $x \in X$ and uniformly for $t \in [0,t_0]$.
\end{thm}
\begin{proof}
First, Theorem \ref{Thm:Chernoff} yields
\begin{equation}
\lim_{n \to \infty}  \norm{T(t)x-V^n(t/n)x}=0 \label{eq:main1}
\end{equation}
for all $x \in X$ and uniformly for $t \in [0,t_0]$, whose proof is based on the Second Trotter-Kato Approximation Theorem  (\cite{engel1999one}, Theorem III.4.9) and the, so-called, $\sqrt{n}$-Lemma (\cite{Chernoff}, Lemma 2). Then, Lemma \ref{Lemma} implies that
\begin{equation}
 \lim_{n \to \infty} \norm{V^n(t/n)x - \prod^n_{i=1} V(a_{n,i}t)x}=0, \label{eq:main2}
\end{equation}
for all $x \in X$ and uniformly for $t \in [0,t_0]$. Since,
\begin{eqnarray}
 \norm{T(t)x- \prod^n_{i=1} V(a_{n,i}t)x} &\leqslant& \norm{T(t)x-V^n(t/n)x} \nonumber \\
 &&+ \norm{V^n(t/n)x - \prod^n_{i=1} V(a_{n,i}t)x}, \nonumber
\end{eqnarray} 
the combination of \eqref{eq:main1} and \eqref{eq:main2} yields \eqref{eq:Chernoff2}.
\end{proof}

As Chernoff's Theorem provides a generalized Lie-Trotter product formula \cite{Trotter}, an immediate application of our theorem yields further Lie-Trotter type product formulas.

\begin{cor}
\label{Cor:Trotter}
 Suppose that $(T_1(t))_{t\geq 0}$ and $(T_2(t))_{t\geq 0}$ are strongly continuous contraction semigroups on a Banach space $X$ with generators $\left(A_1, D(A_1) \right)$ and $\left(A_2, D(A_2) \right)$, 
 respectively. If the assumptions of Theorem \ref{Thm:main} hold, i.e. $V(t)=T_1(t) T_2(t)$ has the aforementioned commutativity property and the sequence $\left(a_{n,i}\right)_{i=1}^n$ behaves like in that Theorem, 
 $A_1+A_2$ is closed, and generates a strongly continuous semigroup $(T_3(t))_{t\geq 0}$, then 
 \begin{eqnarray}
  T_3(t)x &=& \lim_{n\to \infty}  \prod^n_{i=1} T_1(a_{n,i}t) T_2(a_{n,i}t)x, \label{eq:Trotter}
 \end{eqnarray}
for all $x \in X$ and uniformly for $t \in [0,t_0]$.
\end{cor}
\begin{proof}
 If $x \in D(A_1+A_2)=D(A_1)\cap D(A_2)$ we have
\begin{equation}
\lim_{h\to 0} \frac{V(h)x-x}{h}= \lim_{h\to 0} \frac{1}{h} \left[T_1(h) \big(T_2(h)x-x\big)+ T_1(h)x-x \right]=A_2 x+ A_1x. \nonumber  
\end{equation}
Now, the application of Theorem \ref{Thm:main} yields \eqref{eq:Trotter}. 
\end{proof}
Moreover, one can formulate further corollaries, which are deduced from the applications of Chernoff's Theorem, 
like the Post-Widder Inversion Formula (\cite{engel1999one}, Corollary II.5.5). We omit to do this, because these statements work verbatim as Corollary \ref{Cor:Trotter}.  

\section{Examples}

\begin{ex}
\label{ex1}
 Let $\mathcal{H}$ be a Hilbert space and let $L$ be a bounded self-adjoint operator in $\mathcal{H}$. We consider the following contraction in the Banach space $\mathcal{B}_1(\mathcal{H})$ 
 of trace class operators
\begin{equation}
V(t)(X)=\sqrt{\frac{\gamma}{\pi}} \int^{\infty}_{-\infty}\, dy \, e^{-\gamma/2 (y-\sqrt{t}L)^2} X e^{-\gamma/2 (y-\sqrt{t}L)^2}, \quad X \in \mathcal{B}_1(\mathcal{H}), \label{eq:contmeas}    
\end{equation}
with $\gamma>0$. Furthermore, we have $V(0)(X)=X$ and $V(t)V(t')=V(t')V(t)$. In the context of an axiomatic formulation of quantum mechanics, an interesting question is the dynamic of continuous observations with generalized observables 
\cite{Barchielli}. 
If the duration of the whole observation is $t$ and $t/n$ is the time interval between two subsequent measurements, the limit $n\to \infty$ of continuous observation, formulated as
\begin{equation}
 \lim_{n \to \infty} V^n(t/n)(X)=T(t)(X), \nonumber
\end{equation}
exists due to Chernoff's Theorem, where $(T(t))_{t\geq 0}$ is a uniformly continuous semigroup with generator
\begin{equation}
 A(X)= \lim_{h\to 0} \frac{V(h)(X)-X}{h}=-\frac{\gamma}{4} \left(L^2X + XL^2-2LXL\right).
\end{equation}
Now, we split $[0,t]$ into $n$ time intervals with lengths $t_1, t_2, \dots , t_n$ such that $a_{n,i}=t_i/t$. In this example, $A$ is a bounded operator, which generates a semigroup. Therefore, $V(t)$ satisfies \eqref{eq:speccond}, and thus it is enough that the sequence $\left(a_{n,i}\right)_{i=1}^n$  is subject only to the assumption of Proposition \ref{Prop:exp}. Then, the results of Proposition \ref{Prop:Smolyanov} implies 
\begin{equation}
 \lim_{n \to \infty} \prod^n_{i=1} V(a_{n,i}t)(X)=T(t)(X). \nonumber
\end{equation}
\end{ex}

\begin{ex}
\label{ex2}
Chernoff's Theorem has also application in probability theory. It is connected to the central limit theorem, which was established by Goldstein \cite{Goldstein}. Based on the result's of Goldstein, here we show the 
reader a central limit theorem following from Theorem \ref{Thm:main}. We consider $X=C_0(\mathbb{R})$ the space of uniformly continuous functions from $\mathbb{R}$ to $\mathbb{R}$, which is a Banach space with the 
supremum norm. Let $\xi_1$, $\xi_2$, $\dots$ be a sequence of independent identically distributed random variables with probability density functions $p_{\xi}(x) \in C_0(\mathbb{R})$ having two finite moments: $\mathbb{E}(\xi_i)=0$ and 
$\operatorname{Var}(\xi_i)=1$. For a random variable $\xi$ with $p_{\xi}(x) \in C_0(\mathbb{R})$ we define operators $V_\xi(t)$ on $C_0(\mathbb{R})$ by
\begin{equation}
 V_\xi(t)f(x)=\int_\mathbb{R}\,dy\, f(x-\sqrt{t}y) p_{\xi}(y)
\end{equation}
for all $f \in C_0(\mathbb{R})$ and $t>0$. These operators are all contractions, $V_\xi(0)f(x)=f(x)$, and
\begin{equation}
  V_\xi(t)V_\xi(t')f(x)= V_\xi(t')V_\xi(t)f(x). 
\end{equation}
Consider the random variable
\begin{equation}
 \zeta_n=\sum^n_{i=1} \sqrt{a_{n,i}} \xi_i
\end{equation}
where the sequence $\left(a_{n,i}\right)_{i=1}^n$  satisfies the assumptions in Theorem \ref{Thm:main}. One obtains the generator (\cite{Goldstein}, p. 199)
\begin{equation}
 Af(x)= \lim_{h\to 0} \frac{V_\xi(h)f(x)-f(x)}{h}=\frac{1}{2} f''(x) 
\end{equation}
with $D(A)=\{f(x) \in C_0(\mathbb{R}): \text{all}\, f' \, \text{and}\, f'' \,\text{exist and are in}\, C_0(\mathbb{R}) \}$. In fact, $A$ generates a strongly continuous semigroup $T(t)$ defined by
\begin{equation}
 T(t)f(x)=\frac{1}{\sqrt{2 \pi}}\int_\mathbb{R}\,dy\, f(x-\sqrt{t}y) e^{-\frac{y^2}{2}},
\end{equation}
which is called the one-dimensional diffusion semigroup \cite{engel1999one}. However, the condition in \eqref{eq:speccond} is not fulfilled. First, let us rewrite 
\eqref{eq:speccond} in the following form
\begin{equation}
    \lim_{h \to 0} \frac{V_\xi(h)T(t)f(x)-T(h+t)f(x)}{h}=0.
\end{equation}
If we pick 
\begin{equation}
 p_{\xi}(x)= \frac{3\sqrt{3}}{\sqrt{2 \pi}} x^2 e^{-\frac{3x^2}{2}},  
\end{equation}
then the limit is not zero. Hence, we need to use Theorem \ref{Thm:main}, which yields
\begin{equation}
 T(t)f(x) = \lim_{n\to \infty}  \prod^n_{i=1} V_{\xi_i}(a_{n,i}t) f(x), \quad f(x) \in C_0(\mathbb{R}).
\end{equation}
Since
\begin{equation}
 \prod^n_{i=1} V_{\xi_i}(a_{n,i}t) f(x)=\mathbb{E}\left(f(x-\sqrt{t} \zeta_n)\right)=\int_\mathbb{R}\,dy\, f(x-\sqrt{t}y) p_{\zeta_n}(y)
\end{equation}
for all $f \in C_0(\mathbb{R})$ and by using Lemma $1$ in \cite{Goldstein} with $t=1$, we obtain the following convergence in distribution 
\begin{equation}
    \lim_{n\to \infty} p_{\zeta_n}(x)=\frac{1}{\sqrt{2 \pi}} e^{-\frac{x^2}{2}},
\end{equation}
i.e.,
\begin{equation}
 \sum^n_{i=1} \sqrt{a_{n,i}} \xi_i  \xrightarrow{d} \mathcal{N}(0,1).
\end{equation}
\end{ex}

\section*{Acknowledgement}

J.Z.B. acknowledges support from AIDAS-AI, Data Analytics and Scalable Simulation, which is a Joint Virtual Laboratory gathering the Forschungszentrum J\"ulich and the French Alternative Energies and Atomic Energy Commission.

\bibliographystyle{plain}
\nocite{*}
\bibliography{chernoff}

\end{document}